\documentclass[prx,aps,twocolumn,10pt,longbibliography]{revtex4-2}
\usepackage[utf8]{inputenc}
\usepackage{braket}
\usepackage{amsmath}
\usepackage{amssymb}
\usepackage{mathtools}
\usepackage{pifont}

\usepackage{graphicx}
\graphicspath{{Circuits/}{./}}
\usepackage[shortlabels]{enumitem}
\setlist{nosep}
\usepackage{amsthm}
\usepackage[hidelinks,breaklinks]{hyperref}

\newtheoremstyle{customdefi}{\topsep}{0.3\topsep}{}{}{\bfseries}{.}{.5em}{}
\theoremstyle{plain}

\newtheorem{thm}{Theorem}

\newtheorem{cor}{Corollary}
\theoremstyle{customdefi}
\usepackage{algorithm,algorithmicx,algpseudocode}
\usepackage{xcolor}

\usepackage{makecell}
\usepackage{tikz-qtree}

\begin{document}
\author{Christoffer Hindlycke}
\email{christoffer.hindlycke@liu.se}
\author{Jan-{\AA}ke Larsson}
\email{jan-ake.larsson@liu.se}
\affiliation{
Department of Electrical Engineering,
Link\"oping University 
581 83 Link\"oping, SWEDEN%
}
\date{\today}
\title{Single-qubit rotation algorithm with logarithmic Toffoli count and gate depth}

\begin{abstract}
Building generic gates from a restricted gate set is a difficult but important problem, especially in the noisy regime where only a limited set of noise-resistant gates are available, e.g., fault tolerant Clifford gates (generated by Hadamard, Phase, and CNOT gates) and fault tolerant Toffoli gates.
The Toffoli is also often used as a building block of many algorithms and will need to be constructed if not directly available.
This makes Clifford+Toffoli an attractive gate set for building generic gates.
In this paper we give a simple and efficient algorithm for building an approximate single-qubit rotation using only Clifford+Toffoli, which in turn enables any generic single-qubit unitary.
An important difference compared to earlier attempts is that the use of the Toffoli allows us to use simple rounding as opposed to a complicated approximation algorithm.
The resulting gate array does not rely on repeatedly applying a fixed rotation, but immediately applies a rotation $R_{\theta^\ast}$ that is $\epsilon$-close to the desired rotation $R_\theta$, with success probability strictly greater than $1/2$.
It can be rerun upon failure, giving an expected number of repetitions strictly less than 2, an expected Toffoli count strictly less than $4\lceil\log\tfrac{1}{\epsilon}\rceil$, an expected gate depth strictly less than $4\lceil\log\tfrac{1}{\epsilon}\rceil+6$, and uses $2\lceil\log\tfrac1\epsilon\rceil$ ancillas.
The small circuit depth of our construction enables low-noise gates on existing quantum computational devices, and allows for arbitrary precision using a very modest number of ancillas.

\end{abstract}
\maketitle

\section{Introduction}

The Toffoli gate is a fundamental tool in quantum computing and is crucial in Shor's algorithm \citep{Shor1994,Shor1997}, quantum error correction \citep{Shor1995,Cory1998, Knill2001, Chiaverini2004, Reed2012, Nigg2014}, and shallow quantum circuits \citep{Bravyi2018}. 
Many of these quantum algorithms also need arbitrarily small single-qubit rotations.
In general, approximation of generic gates using a finite gate set is essential.
This is especially true in the Fault Tolerant Quantum Computing era where direct fault-tolerant constructions are typically available only for a limited number of gates, so that short generation of generic gates is an absolute necessity for high performance computing \citep{Preskill1998}.
The Solovay-Kitaev theorem \citep{Kitaev1997,Dawson2005} tells us that with access to a limited but well-chosen set of single-qubit quantum gates, an approximation to any desired single-qubit gate to within distance $\epsilon$ can be constructed, with gate depth close to qubic in the logarithm of $\tfrac1\epsilon$; this has recently been improved to the power 1.44 \cite{Kuperberg2023}. 

For specific gate sets, the most common gate set to study is that of Clifford+$T$.
The algorithm of Kliuchnikov \textit{et al.}\ \cite{Kliuchnikov2013} achieves complexity logarithmic in $\frac1\epsilon$, which is optimal up to a multiplicative constant; subsequent efforts have been focused on identifying and reducing this constant \citep{Paetznick2014,Kliuchnikov2016,Ross2016,Selinger2014,Bocharov2015,Gheorghiu2022}.

The Clifford+Toffoli gate set will be available in future quantum computer hardware, since the Toffoli gate is indispensable for quantum computing, and a large amount of work has been put towards its practical implementation \citep{Lanyon2009,Monz2009,DiCarlo2009,Fedorov2012,Levine2019}.
It is therefore natural to ask if generic unitaries can be built using that gate set; some recent results \citep{Amy2023,Mukhopadhyay2024} give complexity $O(\log\frac1\epsilon)$.
An equivalent gate set is Clifford+Controlled-Phase where the conversion cost from one to the other is a factor 3 \cite{Bian2023,Mukhopadhyay2024}, so that e.g., the algorithm from Ref.~\cite{Glaudell2019} gives an equivalent Toffoli count of $24\log\frac1\epsilon$.

Here we give an explicit construction for approximating any $z$-axis rotation using the Clifford+Toffoli gate set with expected Toffoli count $4\lceil\log\frac1\epsilon\rceil$.
Three such rotations allows approximation of any single qubit gate.
The construction is very simple and efficient, as it requires only a trigonometric expression of the desired rotation angle, rounded to the desired accuracy.
The resulting algorithm has expected gate depth $4\lceil\log\frac1\epsilon\rceil+6$ and uses $2\lceil\log\frac1\epsilon\rceil$ ancillas; the low gate depth and algorithm simplicity will make constructions of this kind ubiquitous in quantum computing.

\begin{figure}[hb]
	\centering
	a) \includegraphics[scale=1.1]{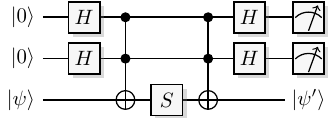}\bigskip
	
	b) \includegraphics[scale=1]{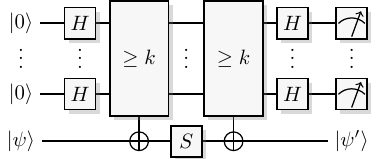}\smallskip
	
	\caption{Circuits for approximate qubit rotation. a) Circuit that applies $R_\varphi$ to $\ket{\psi}$ ($\cos \varphi = 3/5$, $\sin\varphi=4/5$) with probability $5/8$; otherwise a $Z$ gate. b) Circuit that applies $R_{\theta^\ast}$ to $\ket{\psi}$, close to a desired rotation $R_\theta$ with high probability, where ancilla count $n$ and comparison $k$ should be chosen according to Theorem \ref{thm:distance}.}
	\label{fig:qsubroutine}
\end{figure}

Our starting point is an algorithm, conceptually simpler than Solovay-Kitaev, but less efficient: repeatedly applying the fixed rotation generated by a circuit by Nielsen and Chuang \cite[p.~198]{Nielsen2010} that contains only gates from Clifford+Toffoli.
This fixed $\varphi$-rotation around the $z$ axis has $\cos \varphi = 3/5$, $\sin \varphi = 4/5$ and is applied with probability 5/8; see Figure \ref{fig:qsubroutine}a. 
As $\varphi$ is an irrational multiple of $\pi$, repeated application allows for approximation of any rotation on one qubit around the $z$ axis.
This algorithm has constant gate depth for each repetition but requires $O(\frac1\epsilon)$ repetitions on average. 

In this work, we modify this algorithm to be efficient, avoiding a high repetition count.
We introduce additional ancillary controls and replace the two Toffoli gates with $\geq k$ tests in the circuit, giving control over the actual rotation angle.
This gate array either applies the desired rotation to within distance~$\epsilon$ with probability strictly greater than 1/2 or applies a $Z$ gate; the measurement carried out on the controls indicate if a $Z$ gate was applied, and then we correct the input state and try again.
The resulting algorithm has non-constant gate depth $O(\log\frac1\epsilon)$ for each repetition, but instead requires strictly less than two (a constant number of) repetitions on average.
We essentially replace a linear number $O(\frac1\epsilon)$ of repetitions with a logarithmic number $O(\log\frac1\epsilon)$ of ancillas and interconnecting Toffoli gates.

\section{Approximating an arbitrary rotation}
The circuit in Figure \ref{fig:qsubroutine}a works by marking a certain number of superposed states, by adding the relative phase factor $i = e^{i\frac{\pi}{2}}$.
Increasing the number of controls increases the number of superposed states, giving more fine-grained control over the exact ratio of marked/unmarked states and therefore more control over the resulting rotation angle.
Then, the two Toffoli gates can be replaced with an $n$-qubit $\ge k$ test, as depicted in Figure \ref{fig:qsubroutine}b.
A $\ge k$ test on $n$ qubits is possible for $0\le k\le 2^n$: If the control binary digit is greater than or equal to a given limit $k$, the target is inverted.
We focus on the quantum state after the second $\ge k$ test; of the computational basis states contained in $\ket{+^n,0}$, the circuit leaves $k$ states unmarked and marks $2^n-k$ states with an extra phase factor $i=e^{i\frac\pi2}$.
With 
\begin{equation}
\begin{split}
C&=|k+i(2^n-k)|\\
0\le v&=\arg(k+i(2^n-k))\le \tfrac\pi2,
\end{split}
\label{eq:v}
\end{equation}
the scalar product with the $\ket{+^n,0}$ state is
\begin{equation}
\frac{k+i(2^n-k)}{2^n}
=\frac{C\cos v+iC\sin v}{C(\cos v+\sin v)}
=\frac{e^{iv}}{\cos v+\sin v}.
\end{equation}
For the $\ket{+^n,1}$ state, the circuit marks $k$ states of the computational basis and leaves $2^n-k$ states unmarked. 
The scalar product with the $\ket{+^n,1}$ state is
\begin{equation}
\hspace{-1mm}\frac{2^n-k+ik}{2^n}
		=\frac{e^{i\left(\frac\pi2-v\right)}}{\cos v+\sin v}
		=\frac{e^{iv}}{\cos v+\sin v}e^{i\left(\frac\pi2-2v\right)}.
\end{equation}
Therefore, if the measurement outcome after the final Hadamards is $n$ zeros, we have performed an $R_{\theta^\ast}$ rotation, where $\theta^\ast=\frac\pi2-2v$, so that $-\tfrac\pi2\le\theta^\ast\le\tfrac\pi2$ and
\begin{equation}
		k=C\cos v=2^n\frac{\cos v}{\cos v+\sin v}
		=2^n\frac {1+\tan\tfrac{\theta^\ast}2}2.
  \label{eq:b}
\end{equation}
We can now perform a rotation $R_{\theta^\ast}$ arbitrarily close to an ideal rotation $R_\theta$ for $-\tfrac\pi2\le\theta\le\tfrac\pi2$, by adjusting the number of controlling ancillas $n$ and rounding $2^{n-1}\tan\frac\theta2$ to the closest integer.
Calculating $\tan\tfrac\theta2$ can be done in polynomial time on a classical computer.
We have the following theorem.\smallskip
\begin{thm}
	\label{thm:distance}
	Given a desired rotation angle $-\tfrac\pi2\le\theta\le\tfrac\pi2$ and accuracy $\epsilon>0$, choose ancilla count 
	\begin{equation}
  n=1+\Bigl\lceil\log\tfrac1\epsilon\Bigr\rceil
    \label{eq:n}
	\end{equation}
	and comparison
  	\begin{equation}
  k=2^{n-1}+\Bigl\lfloor2^{n-1}\tan\tfrac{\theta}{2}+\tfrac12\Bigr\rfloor,
\end{equation}
  this corresponds to choosing $\theta^\ast$ so that
  \begin{equation}
  	\label{eq:theta-ast}
  	\begin{split}
  		 2^{n-1}\tan\tfrac{\theta^\ast}{2}
      =\Bigl\lfloor2^{n-1}\tan\tfrac{\theta}{2}+\tfrac12\Bigr\rfloor.
  	\end{split}
  \end{equation}
  Using $\ket{0}^{\otimes n}
  \ket{\psi}$ as input to the circuit in Figure  \ref{fig:qsubroutine}b, if the ancilla measurement outcomes are all zero, the circuit applied an $R_{\theta^\ast}$ rotation to $\ket{\psi}$, with an extra global phase, such that $||R_\theta- R_{\theta^\ast}|| \leq \lvert \theta - \theta^\ast \rvert \leq \epsilon$. 
  This happens with probability $\tfrac12 + \tfrac12\tan^2 (\frac{\theta^\ast}{2})\ge\tfrac12$.
\end{thm}
\begin{proof}
	By the preceding equations, when the measurement outcomes are all zero the circuit applied the rotation  $R_{\theta^\ast}$ with an extra global phase $v=\tfrac\pi4-\tfrac{\theta^*}2$, and this happens for all $\ket\psi$ with probability         
\begin{equation}
	P(0^n)=\frac1{(\cos v+\sin v)^2}
  =\frac1{1+\cos\theta^\ast}
	=\frac{1 + \tan^2 \frac{\theta^\ast}{2}}{2}.
\end{equation}
	It remains to bound the distance between the actual rotation $R_{\theta^\ast}$ being carried out and the desired rotation $R_\theta$.
	We have
	\begin{equation}
		\begin{split}
			||&R_\theta-R_{\theta^*}||
			=\smashoperator{\sup_{\braket{\varphi \vert \varphi} = 1}}\bigl\lVert (R_\theta-R_{\theta^\ast})\ket\varphi
			\bigr\rVert_2
      =\bigl|e^{i\theta} - e^{i \theta^*}\bigr|\\
			&=\bigl|e^{i\frac{\theta+\theta^*}2}
			\bigl(e^{i\frac{\theta-\theta^*}2}-e^{-i\frac{\theta-\theta^*}2}\bigr)\bigr|
			=2\bigl|\sin\tfrac{\theta-\theta^\ast}2\bigr|\\
      &\le|\theta-\theta^*|
      \le2\bigl|\tan\tfrac{\theta-\theta^*}2\bigr|
      =2\biggl|\frac{\tan\tfrac\theta2 - \tan\tfrac{\theta^\ast}{2}}
      {1+\tan\tfrac\theta2\tan\tfrac{\theta^\ast}{2}}\biggr|\\
      &\le2\bigl|\tan\tfrac\theta2-\tan\tfrac{\theta^*}2\bigr|
      \le 2^{-n+1}\le
       2^{-\log\frac1\epsilon}=
       \epsilon,
		\end{split}
	\end{equation}
because $\tan\tfrac\theta2\tan\tfrac{\theta^\ast}2\ge0$.
\end{proof}

In Theorem~\ref{thm:distance} we use the smallest possible choice for~$n$ such that the error is less than $\epsilon$.
We note that for $\theta^\ast \neq 0$, the probability of success is strictly greater than $1/2$.
We should also point out that if $\theta^\ast=0$ measurement outcomes all zero will result in an $R_0=\mathbb I$ rotation having been applied by our circuit, which then must be close enough to the desired $R_\theta$ rotation.
This case is better handled by applying the identity operator to $\ket{\psi}$, the success probability is then~1.
Similarly, $\theta^\ast=\pm\tfrac\pi2$ correspond to applying the $S=R_{\pi/2}$ or the $S^\dag=R_{-\pi/2}$ rotation, but also gives the trivial comparisons $k=2^n$ or $k=0$; the $S$ or $S^\dag$ gate is then applied with probability 1, as stated in the theorem.\smallskip
\begin{cor}
	The approximation in Theorem \ref{thm:distance} can be made exact iff there exists an $n$ such that $2^n \tan \tfrac{\theta}{2}$ is an integer.
\end{cor}
\begin{proof}
	Theorem \ref{thm:distance} gives $\theta^* = \theta$ iff using such an~$n$.
\end{proof}

\begin{cor}
	The set of $\theta$ for which Theorem~\ref{thm:distance} gives an exact $R_\theta$ is dense in $-\tfrac\pi2\le\theta\le\tfrac\pi2$.
\end{cor}
\begin{proof}
  Let $\theta_m=2\arctan( 2^{-m}\lfloor2^m\tan\tfrac\theta2+\tfrac12\rfloor)$. 
  Theorem~\ref{thm:distance} gives $\theta_m^*=\theta_m$ for $n\ge m+1$, and $\smashoperator{\lim\limits_{m\to\infty}}\theta_m=\theta$\rule{0pt}{1em}.
\end{proof}

If the measurement outcomes are not all zero, the following holds.
\begin{thm}
	\label{thm:zgate}
	Let $n,k,\theta^*$ be as in Theorem \ref{thm:distance}. 
  Using $\ket{0}^{\otimes n} \ket{\psi}$ as input to the circuit in Figure  \ref{fig:qsubroutine}b, if the ancilla measurement outcomes are not all zero, the circuit applied a $Z$ gate to $\ket{\psi}$, with an extra global phase. 
  This happens with probability $\tfrac12 - \tfrac12\tan^2 (\frac{\theta^\ast}{2})\le\tfrac12$.
\end{thm}
\begin{proof}
	The probability that the measurement outcomes are not all zero follows from Theorem~\ref{thm:distance}. 
 	After the second $\geq k$ test, the measurement eigenstate for outcome $x$ is
	\begin{equation}
    \ket{\Theta_{n,x}} = H^{\otimes n}\ket x
    =2^{-\frac n2}\sum_{y=1}^{2^n-1}(-1)^{x\cdot y}\ket y.
	\end{equation}
  Also, with $\ket\psi=\ket0$, at that point the quantum state is
	\begin{equation}
    		\ket{\Psi_{n,k}} = 
     		2^{-\frac n2} \biggl(\sum_{j = 0}^{k-1} \ket{j} 
  			+ \sum_{j=k}^{2^n-1} i\ket{j}\biggr),
 	\end{equation}
  and with $\ket\psi=\ket1$ the quantum state is
  \begin{equation}
  	\ket{\Phi_{n,k}} = 
    			2^{-\frac n2} \biggl(\sum_{j = 0}^{k-1} i\ket{j} 
    			+ \sum_{j=k}^{2^n-1} \ket{j}\biggr).
  \end{equation}
  If $x\neq 0$ we have
  \begin{equation}
  \begin{split}
    &\langle\Theta_{n,x}|\Psi_{n,k}\rangle
    =2^{-n}\biggl(\sum_{j=0}^{k-1}(-1)^{x\cdot j}+i\sum_{j=k}^{2^n-1}(-1)^{x\cdot j}\biggr)\\
    &\qquad=2^{-n}\biggl((1-i)\sum_{j=0}^{k-1}(-1)^{x\cdot j}+i\sum_{j=0}^{2^n-1}(-1)^{x\cdot j}\biggr)\\
    &\qquad=2^{-n}(1-i)\sum_{j=0}^{k-1}(-1)^{x\cdot j}
  \end{split}
  \label{eq:Z}
  \end{equation}
  and
  \begin{equation}
  \begin{split}
    &\langle\Theta_{n,x}|\Phi_{n,k}\rangle
    =2^{-n}\biggl(i\sum_{j=0}^{k-1}(-1)^{x\cdot j}+\sum_{j=k}^{2^n-1}(-1)^{x\cdot j}\biggr)\\
    &\quad=2^{-n}(i-1)\sum_{j=0}^{k-1}(-1)^{x\cdot j}
    =-\langle\Theta_{n,x}|\Psi_{n,k}\rangle
  \end{split}
  \end{equation}
  so the circuit applied a $Z$ gate on $\ket\psi$ with the global phase $\arg(1-i)=-\tfrac\pi4$ or $\arg(i-1)=\tfrac{3\pi}4$ depending on the sign of the remaining sum in Eqn.~(\ref{eq:Z}).
\end{proof}

\section{Explicit gate array and gate counts}
A straightforward construction of the $\ge k$ test is to encode $k$ in an additional ancillary $n$-qubit register and use a ripple-carry comparator \citep{Cuccaro2004}.
But since $k$ is a classical constant, we can also use it to build the gate array. 
A ripple-carry bitwise comparison with a classical constant $k$ is shown in Figure~\ref{fig:comp}a-c. 
Note that the circuit in Figure \ref{fig:qsubroutine}a is a special case of our circuit where $n = 2$ and $k = 3$.
Also, the construction is somewhat similar to a textbook $n$-control Toffoli construction, in fact, using this design for a $\ge k=2^n-1$ comparator gives exactly the construction in Ref.~\cite{Nielsen2010} Figure~4.10.

This greater-than ($>k$) comparator can be turned into a greater-or-equal ($\ge k$) comparator by setting carry-in to 1, such a constant input can also be used to simplify the gate array.
If $k_j=0$ and $\ket c$ is known to be $\ket1$ (see Figure~\ref{fig:comp}d), neither carry-in  nor the input $\ket x$ are used so can be omitted from the circuit, and carry-out is $\ket1$ so further reductions will occur.
If $k_j=1$ and $\ket c$ is known to be $\ket1$ (see Figure~\ref{fig:comp}e), the comparison outputs $\ket x$ itself so the carry-in can be omitted; if the bit comparison target is an ancillary system also the CNOT and target ancilla can be omitted, only $\ket x$ needs to be retained. 

With these simplifications the lowest significant bit comparison uses no gates and no ancilla internal to the comparator.
Also, the highest significant comparison uses a single Toffoli and no ancilla internal to the comparator, and the others use two Toffolis and one ancilla internal to the comparator.
Thus, the comparator uses $2n-3$ Toffolis, and $n-2$ internal ancillas.

We use two comparators in the circuit which would give total Toffoli count $4n-6$ but this can be reduced further by omitting the uncomputation of the intermediate values in the internal ancillas and reusing the same ancillas in the second comparator, which then only needs to perform the uncomputation step.
After this reduction, the Toffoli count is $2n-2=2\lceil\log\frac1\epsilon\rceil$.
Since the comparator-internal ancilla count is $n-2$, the total ancilla count of the gate array is also $2n-2=2\lceil\log\frac1\epsilon\rceil$.

If the lowest significant bits of $k$ are 0, the Toffoli count reduces further because of the constant output of the bit comparator of Figure~\ref{fig:comp}d, in which case the lowest significant bits of $x$ are also not needed.
When this happens one can reduce the ancilla count $n$ but retain the same accuracy $\epsilon$, we give an example below that may clarify this. 
In any case, an upper bound to the Toffoli count is $2n-2=2\lceil\log\frac1\epsilon\rceil$.

\begin{figure}
	\includegraphics[width=\linewidth]{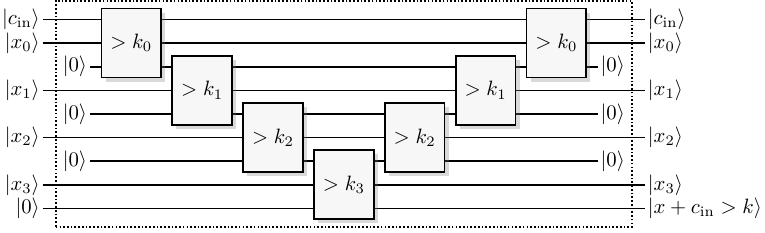}
  a) Comparator circuit from bitwise ripple-carry comparisons\hfill\medskip\\
  \includegraphics{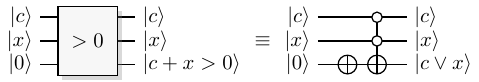}\\
  b) Bit comparison with $k_j=0$, white controls are inverted\hfill\medskip\\
  \includegraphics{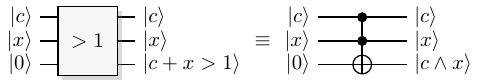}\\
  c) Bit comparison with $k_j=1$\hfill\medskip\\
  \includegraphics{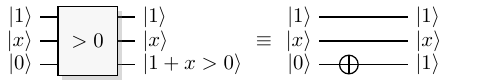}\\
  d) Bit comparison with $k_j=0$ and carry-in $\ket c=\ket1$\hfill\medskip\\
  \includegraphics{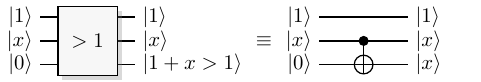}\\
  e) Bit comparison with $k_j=1$ and carry-in $\ket c=\ket1$\hfill\medskip\null
	\caption{Bitwise ripple-carry comparator made of a sequence of Toffoli gates. For indices where $k_i=0$, the modified Toffoli realizes a quantum OR gate. If carry-in is constant 1 the bit comparisons simplify.}
	\label{fig:comp}
\end{figure}

\section{Repeat until success}
The circuit of Figure \ref{fig:qsubroutine}b performs a probabilistic approximation, so can fail. 
The following algorithm, which has random runtime, instead always succeeds and has several other good properties.
\begin{algorithm}[H]
	\caption{Rotation algorithm}\label{alg:rotation}
	\begin{algorithmic}[1]
		\Statex \hspace{-5mm}\textbf{Input} $\theta, \epsilon,\ket{\psi}$
		\Statex \hspace{-5mm}\textbf{Output} $R_{\theta^\ast} \ket{\psi}$ where $||R_\theta- R_{\theta^\ast}|| \leq \epsilon$
		\item\textbf{choose} $n, k, \theta^*$ as in Theorem \ref{thm:distance}
		\item\textbf{if} $\theta^* \neq 0$ \textbf{repeat}
		\item\quad\textbf{prepare} $n$ ancillary qubits $\ket{\phi} = \ket{0}^{\otimes n}$
		\item\quad\textbf{apply} $H^{\otimes n}$ to $\ket{\phi}$
		\item\quad\textbf{apply} a $\geq k$ test on $\ket{\phi}$ with target $\ket{\psi}$
		\item\quad\textbf{apply} $S$ to $\ket{\psi}$
  	\item\quad\textbf{apply} a $\geq k$ test on $\ket{\phi}$ with target $\ket{\psi}$
		\item\quad\textbf{apply} $H^{\otimes n}$ to $\ket{\phi}$
		\item\quad\textbf{measure and discard} ancillas $\ket\phi$ 
		\item\quad\textbf{if} outcomes are not all zero \textbf{then}
		\item\quad\quad\textbf{apply} $Z$ to $\ket{\psi}$
		\item\textbf{until} outcomes are all zero
	\end{algorithmic}
\end{algorithm}

\begin{thm}
	\label{thm:toffcount}
	Let $X$ be the random number of times Algorithm~\ref{alg:rotation} repeats until it applies an $R_{\theta^\ast}$ rotation to $\ket{\psi}$. 
	Then $X$ is geometrically distributed, $X \geq 1$, $P(X>m)<2^{-m}$ and $E(X) < 2$.
\end{thm}
\begin{proof}
	Each repetition of Algorithm~\ref{alg:rotation} is independent and at least one attempt is needed, so $X \geq 1$.
	By Theorem \ref{thm:zgate} the probability of applying ${R_{\theta^\ast}}$ to $\ket\psi$ is strictly greater than $1/2$.
	Then, $X$ is geometrically distributed, $P(X>m)<1-\sum_{j=1}^{m}2^{-j}=2^{-m}$
   and $\text{E}\left(X\right) < \tfrac{1}{{1}/{2}} = 2$.
\end{proof}
\begin{cor}
	Algorithm~\ref{alg:rotation} has expected Toffoli count strictly less than $4\lceil\log\tfrac{1}{\epsilon}\rceil$.
\end{cor}
\begin{proof}
  The gate array has at most Toffoli count $2\lceil\log\tfrac{1}{\epsilon}\rceil$, so
	the corollary follows from Theorem \ref{thm:toffcount}.
\end{proof}
\begin{cor}
  Algorithm~\ref{alg:rotation} has expected gate depth strictly less than $4\lceil\log\tfrac{1}{\epsilon}\rceil+6$.
\end{cor}
\begin{proof}
  If the gate in Figure \ref{fig:comp}b is counted as a single modified Toffoli, the gate depth is the Toffoli count plus 3, so the corollary follows from Theorem \ref{thm:toffcount}.
\end{proof}

\section{Example: approximating the $T$ gate}
The $T$ (or $\tfrac{\pi}{8}$) gate is a frequently used universal gate that applies an $R_{\tfrac\pi4}$ rotation to a single qubit.
For illustrative purposes we perform a coarse approximation here with $||T- R_{\theta^\ast}||\le|\tfrac\pi4-\theta^\ast|\le 10^{-2}$. 
The number of ancillas should be chosen as 
\begin{equation}
n=1+\lceil\log10^2\rceil=8,
\end{equation}
so that
\begin{equation}
	\begin{split}
		k&=2^7 + \bigl\lfloor 2^7 \tan \tfrac{\pi}{8} + \tfrac{1}{2} \bigr\rfloor
		= 128 + 53
		=10110101_2.
	\end{split}
\end{equation}
The probability of success per repetition is
\begin{equation}
	\begin{split}
		P\left(0^n\right)
		= \frac{1 + \left(\tfrac{53}{128}\right)^2}{2}
    		=\frac{19193}{32768}
		\approx
		0.5857.
	\end{split}
\end{equation}
We select $n$ to give at least the desired accuracy, but here
\begin{equation}
  2^7\tan\tfrac\pi8=11 0101.000001001\ldots_2,
\end{equation}
so the rounding error is much smaller than $\tfrac12$.
The actual approximation error is
\begin{equation}
||T-R_{\theta^\ast}||\le\bigl|\tfrac\pi4-\theta^\ast\bigr|=\tfrac\pi4-2\arctan\tfrac{53}{128}\approx2.579\times10^{-4}.
\end{equation}
The small rounding error has the effect that selecting $n$ in the range $8<n\le12$ would give zeros for the lowest significant bits of $k$. 
\begin{figure}[tbp]
	\includegraphics[width=\linewidth]{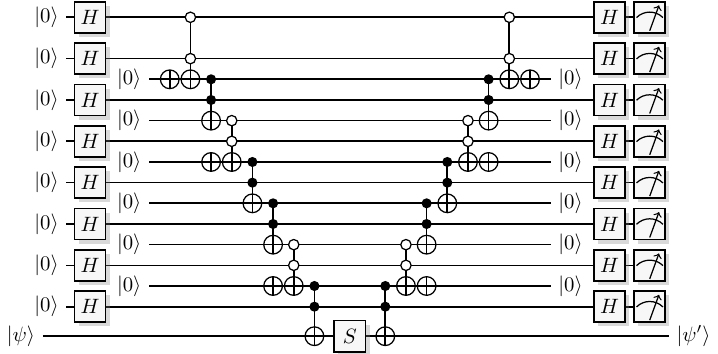}
	\caption{Gate array that approximates the $T$ gate. If ancilla measurement outcomes are all zero, the circuit applied an $R_{\theta^\ast}$ rotation to $\ket\psi$ where $||T - R_{\theta^\ast}||\le|\tfrac\pi4-\theta^\ast|<2.6\times10^{-4}$. This happens with probability $> 0.58$; for other measurement outcomes the circuit applied a $Z$ gate on $\ket\psi$. The circuit contains a comparator with the integer $k=10110101_2$, starting with the least significant bit comparison at the top.}
	\label{fig:T}
\end{figure}
Then the simplification in Figure~\ref{fig:comp}d would apply one or more times, so that the circuit still would simplify to that in Figure~\ref{fig:T}.
For the $T$ gate, the choice $n=8$ gives such a good approximation that the next improvement does not occur until $n=13$.

In actual use the accuracy will need to be better. 
For example, factoring an $m$-bit integer using Shor's algorithm requires approximating close to $2m\log m$ single-qubit gates \cite{Kliuchnikov2016}, and with $m=10000$ this number is $2.7\times10^5$.
If the total error due to gate approximation is required to be less than 0.01\%, assuming the errors add, we need $\epsilon\approx 10^{-4}/(2.7\times10^5)\approx 3.8\times10^{-9}$ or~$n=33$.

\section{Conclusions}
We have proposed an algorithm for generating a phase rotation $R_{\theta^*}$ that is $\epsilon$-close to a desired rotation $R_\theta$, inspired by an earlier construction in Nielsen and Chuang \cite{Nielsen2010}.
The construction uses a $\geq k$ test that allows for simple selection of the rotation angle, and consists only of gates from the Clifford+Toffoli gate set.
Our algorithm that uses this circuit has expected Toffoli count strictly less than $4\lceil\log\tfrac1\epsilon\rceil$, expected gate depth strictly less than $4\lceil\log\tfrac1\epsilon\rceil+6$, and uses $2\lceil\log\tfrac1\epsilon\rceil$ ancillas,

This is a cubic improvement over the Solovay-Kitaev algorithm \citep{Kitaev1997, Dawson2005} and an improvement both over existing generic algorithms and those specifically involving the Clifford+Toffoli gate set \citep{Kuperberg2023, Amy2023, Mukhopadhyay2024}.
A lower bound can be obtained by a simple counting argument \cite{Harrow2002}, but that bound depends on the gate set size which here involves the number of ancillas, so the bound increases slightly slower than logarithmic in $\tfrac1\epsilon$, allowing a small margin for improvement. 
However, the present algorithm uses only gates with bounded fan-in \citep{Bravyi2018}, so there may not be such a margin: We conjecture this is the case.
Proving this would be an interesting development of our results.

Our construction only relies on being able to calculate $\tan\tfrac\theta2$, so requires only polynomial classical runtime and space to generate.
In addition, the relative simplicity of our algorithm makes it attractive as a standard tool in quantum engineering.
It is our hope that this, in combination with its efficient implementation, will make it useful for both theoretical and practical pursuits.

\bibliography{lib}

\end{document}